\documentclass[a4paper,UKenglish,numberwithinsect,cleveref, autoref, thm-restate]{lipics-v2021}
\usepackage[utf8]{inputenc}
\usepackage{amssymb}
\usepackage{amsmath}
\usepackage{amsthm}
\usepackage{stmaryrd}
\usepackage{bm}
\usepackage{complexity}
\usepackage[textsize=scriptsize]{todonotes}
\usepackage{booktabs}
\usepackage{tikz}
\usetikzlibrary{positioning, arrows}
\usetikzlibrary{patterns}
\usetikzlibrary{shapes}
\usetikzlibrary{decorations.pathreplacing}
\usetikzlibrary{fadings}
\usetikzlibrary{intersections}
\usepackage{xifthen}
\usepackage{pgfmath}
\usepackage{algorithm,algorithmic}
\usepackage[capitalise]{cleveref}
\usepackage{colortbl}

\nolinenumbers

\newcommand{\inv}{\bm{i}}
\newcommand{\h}{\bm{h}}
\newcommand{\w}{\bm{w}}
\renewcommand{\o}{\bm{o}}

\newcommand{\eqdef}{\stackrel{\text{def}}{=}}

\newcommand{\leqaug}{\leq_{\text{aug}}}
\newcommand{\leqstruct}{\leq_{\text{st}}}
\newcommand{\geqaug}{\geq_{\text{aug}}}
\newcommand{\geqstruct}{\geq_{\text{st}}}
\newcommand{\refl}{{\hookrightarrow}}
\newcommand{\Bad}{{\textit{Bad}}}
\newcommand{\Dec}{{\textit{Dec}}}
\newcommand{\Inc}{{\textit{Ant}}}

\newcommand{\om}{{\omega}}
\newcommand{\omom}{{\om^{\om}}}

\newcommand{\Mul}{M^{\diamond}}
\newcommand{\Mulr}{M^{r}}

\newcommand{\set}[1]{\{\; #1 \;\}}
\newcommand{\setof}[2]{\{\; #1 \;:\; #2 \;\}}
\newcommand{\mset}[1]{\langle #1 \rangle}
\newcommand{\uperp}{\underline{\perp}}

\newcommand{\str}{\textit{str}}
\newcommand{\mot}{m.o.t. }
\newcommand{\pmot}{\operatorname{\underline{\bold{o}}}}

\definecolor{Gray}{rgb}{0.8,0.8,0.8}

\title{Ordinal measures of the set of finite multisets}
\author{Isa {Vialard}}{Université Paris-Saclay, CNRS, ENS Paris-Saclay, Laboratoire Méthodes Formelles, 91190, Gif-sur-Yvette France }{}{https://orcid.org/0000-0002-7261-9342}{}

\authorrunning{I. Vialard}
\Copyright{Isa Vialard}



\theoremstyle{plain}
\theoremstyle{definition}

\begin{document}
\maketitle

\begin{abstract}
Well-partial orders, and the ordinal invariants used to measure them, are relevant in set theory, program verification, proof theory and many other areas of computer science and mathematics.
In this article we focus on one of the most common data structure in programming, the finite multiset of some wpo. There are two natural orders one can define on the set of finite multisets $M(X)$ of a partial order $X$: the multiset embedding and the multiset ordering, for which $M(X)$ remains a wpo when $X$ is. Though the maximal order type of these orders is already known, the other ordinal invariants remain mostly unknown.
Our main contributions are expressions to compute compositionally the width of the multiset embedding and the height of the multiset ordering. Furthermore, we provide a new ordinal invariant useful for characterizing the width of the multiset ordering.
\end{abstract}

\section*{Introduction}

\textbf{Measuring partial orders} intervene in many domains, from set theory to proof theory, including infinitary combinatorics, program verification, rewriting theory, proof automation and many more. 

There are intuitive notions of measure for a partial order when it is finite: its cardinal obviously, but also its height (the length of a maximal chain) or its width (the length of a maximal antichain).
Similar notions exist for infinite partial orders, as long as they are \emph{well partial orders} (wpo), i.e., well-founded partial orders with no infinite antichains \cite{higman52,kruskal60}. 
Two such notions are the \emph{height}, which is the order type of a maximal chain, and the \emph{maximal order type} (m.o.t.), which is the order type of a maximal linearisation, a notion introduced by De Jongh and Parikh in order to measure hierarchies of functions \cite{dejongh77}. These are transfinite measures, hence we call them \emph{ordinal invariants}.
K\v{r}\'i\v{z} and Thomas introduced equivalent characterizations for \mot and height, which naturally led to the definition of a third ordinal invariant, \emph{width} \cite{kriz90b}. Less studied than its counterparts, the width of a wpo relates to its antichains, even though it cannot be defined as the order type of a maximal antichain. While exploring techniques for program termination, Blass and Gurevich developed these characterizations into a game-theoretical point of view which gives useful intuitions \cite{blass2008}.

Ordinal invariants have proven helpful to establish program termination or to analyze well-structured transition systems (wsts) \cite{bonnet2013}, i.e., transition systems whose set of configurations is a wpo and whose transitions respect this ordering. The study of wsts through controlled bad sequences \cite{FFSS-lics2011}, and sometimes controlled antichains  \cite{schmitz2019b},  shows obvious connections between the complexity of  verification problems  on wsts and the \mot and width of the underlying wpo.

\textbf{Computing ordinal invariants compositionally} is paramount, as most wpos one might want to measure are built from classical operations on simpler wpos whose invariants are known.
De Jongh and Parikh computed the \mot of the disjoint sum and the cartesian product of wpos \cite{dejongh77}. Schmidt then computed the \mot of word embedding and homeomorphic tree embedding on a wpo \cite{schmidt2020}. Abraham and Bonnet pursued this line of study by computing the height of cartesian product, but also the width of disjoint sum and lexicographic product \cite{abbo99}.
For a complete survey of these results, see \cite{dzamonja2020}, where ~D{\v{z}}amonja et al. computed the ordinal invariants of the lexicographic product, but also the height of the multiset word and tree embeddings.

\textbf{Finite multisets:} In this article, we study the ordinal invariants of the set of finite multisets. Multisets are one of computer science's most common data structures, especially in rewriting theory. Informally, a finite multiset over a set $X$ is a finite subset of $X$ where an element can appear finitely many times. 
One can see the set of finite multisets on a wpo as the set of finite words quotiented by the equivalence relation ``equality up to some permutation''. It comes down to describing a multiset as a word where the order of terms is irrelevant.
A finite multiset can be represented by a function from $X$ to $\mathbb{N}$ with finite support, which associates each element with its multiplicity.

Two orderings are classically defined on the finite multisets of any ordered set.
The first one is the \emph{multiset ordering}. The intuition behind this ordering is that, from a multiset $m$, one can build a smaller multiset by removing an element $x$ from $m$, and replacing it with an arbitrary number of elements smaller than $x$. The multiset ordering often appears in rewriting theory and automation of termination proofs \cite{dershowitz79}. 

A perhaps more natural, but less known, ordering on the set of finite multisets is the \emph{multiset embedding}, 
or term ordering as it is called in \cite{vandermeeren2015}. It was presented by Aschenbrenner and Pong as a natural extension of the embedding order over finite words \cite{aschenbrenner2004}.

Some invariants of these two orderings have already been measured:
Van der Meeren, Rathjen, and Weiermann \cite{vandermeeren2015} built on \cite{weiermann91} to compute the m.o.t. of the set of finite multisets on a wpo $X$ ordered with the multiset ordering, and provided a new proof for the expression of the m.o.t of the multiset embedding computed in \cite{weiermann2009}. ~D{\v{z}}amonja et al. \cite{dzamonja2020} proved that the height of the multiset embedding is equal to the height of the set of finite words ordered with word embedding. It is noteworthy that these three results give expressions that are functional in the \mot and height of $X$.
However, the height of invariants of the multiset ordering still needs to be determined, and the width remains unstudied for both orderings.

\textbf{Our contributions:} In this article, we provide functional expressions for the width of the multiset embedding (\cref{thm-w-Mul}) and the height of the multiset ordering (\cref{thm-h-Mulr}).

We further show that the width of the multiset ordering is not functional in any of the three ordinal invariants (\cref{ex-w-Mulr-non-functional}). Nonetheless, we get around this issue by introducing a fourth ordinal invariant, the \emph{maximal safe order type} (\cref {def-pmot}), in which the width is functional (\cref{thm-w-Mulr}).

\cref{tab:contribution} sums up this article's contributions (in the gray cases) amidst the current state of the art.

\begin{table}[h]
\caption{Ordinal invariants of $\Mul(X)$ and $\Mulr(X)$.}
	\begin{center}
		\renewcommand{\arraystretch}{1.5}
		\begin{tabular}{r||cc}
			\hline
			\textbf{Invariants}	& $\Mul(X)$ 	& $\Mulr(X)$  	\\
			\hline
			M.o.t. $\o$   	& $\om^{\widehat{\o(X)}} $  & $\om^{\o(X)}$	\\
			Height $\h$     & $\h^*(X)$ 	& \cellcolor{Gray}$\om^{\h(X)}$	\\
			Width $\w$    	& \cellcolor{Gray}$\om^{\widehat{\o(X)}-1}$ & \cellcolor{Gray}$\om^{\pmot(X)}$	\\
			\hline
		\end{tabular}
	\end{center}
	\label{tab:contribution}
\end{table}

\section{Definitions and state of the art}

\subsection{Width, height and maximal order type}

We suppose that the reader is familiar with ordinal arithmetic. Otherwise, the notions used in this article can be found in the appendix.

A sequence $x_1,\dots,x_n,\dots$ on a partial order is \emph{good} when there exist $i<j$ such that $x_i\leq x_j$, otherwise it is a \emph{bad} sequence. A sequence is an \emph{antichain} when all its elements are pairwise incomparable (i.e. $x_i \perp x_j$ for any $i<j$).

A \emph{well partial order (wpo)} is a partial order that has no infinite bad sequences. Equivalently, a wpo is a partial order that is both well-founded (i.e. no infinite strictly decreasing sequences) and has no infinite antichains (i.e. sequences of pairwise incomparable elements).

Let $X$ be a wpo.
We define $\Bad(X)$, $\Dec(X)$ and $\Inc(X)$ as the trees of bad sequences, strictly decreasing sequences, and antichains of $X$, respectively, ordered by inverse prefix order (a sequence is smaller than its prefixes)
The finiteness of bad sequences, strictly decreasing sequences and antichains in a wpo implies that these trees are well-founded. Therefore, one can define a notion of rank: a sequence has rank $0$ when it cannot be extended; otherwise its rank is the smallest ordinal strictly larger than the ranks of smaller sequences (by inverse prefix order).

The \emph{maximal order type} (or m.o.t.) of $X$, denoted with $\o(X)$, is defined as the rank of the empty sequence in $\Bad(X)$.
Similarly, the \emph{height} $\h(X)$ and the \emph{width} $\w(X)$ of $X$ are defined as the ranks of the empty sequence in $\Dec(X)$ and $\Inc(X)$, respectively.
Together, $\o(X)$, $\h(X)$ and $\w(X)$ are called the \emph{ordinal invariants} of $X$.

Since the rank of the empty sequence is the smallest ordinal strictly larger than the ranks of the sequences of length $1$, these definitions can be reformulated inductively through the following \emph{residual equations}:
\begin{align}
\label{Reso}\tag{Res-o}\o(X)&=\sup_{x\in X} \left(\o(X_{\not\geq x})+1\right)\\
\label{Resh}\tag{Res-h}\h(X)&=\sup_{x\in X} \left(\h(X_{< x})+1\right)\\
\label{Resw}\tag{Res-w}\w(X)&=\sup_{x\in X} \left(\w(X_{\perp x})+1\right)
\end{align}

$X_{*x}\eqdef \setof{y\in X}{y * x}$ for $*\in\set{\not\geq,<,\perp}$ is what is called a residual of $X$.

Since $\Dec(X)$ and $\Inc(X)$ are embedded in $\Bad(X)$, we know $\h(X)\leq\o(X)$ and $\w(X)\leq\o(X)$. K\v{r}\'i\v{z} and Thomas found that the m.o.t can also be bounded by the width and the height \cite{kriz90b}:
\begin{theorem}[Height-Width Theorem]
	\label{kt}
	For any wpo $A$, \begin{itemize}
		\item $\o(A) \leq \h(A) \otimes \w(A)$.
		\item Furthermore, if $\o(A)$ is multiplicatively indecomposable (see Appendix) and $\h(A)<\o(A)$ (resp. $\w(A)<\o(A)$) then $\w(A)=\o(A)$ (resp. $\h(A)=\o(A)$).
	\end{itemize}
\end{theorem}

\subsection{Ordinal invariants of basic data structures}

In computer science, sums and products of wpos are the most basic data structures one can find. Their ordinal invariants are easy to compute compositionally, with the notable exception of the width of the cartesian product which is not functional \cite{arxiv2202.07487}. In this article, these structures will often be used as examples, or appear in proofs.

\begin{lemma}[Disjoint sum \cite{dejongh77,abbo99,dzamonja2020}]
\label{def-sqcup}
Let $P,Q$ be two wpos. The \emph{disjoint sum} $P\sqcup Q$ is the disjoint union of $P$ and $Q$ ordered such that for all $p \in P, q\in Q$, $p\perp q$.
It is a wpo and: \begin{itemize}
\item $o(P\sqcup Q) = \o(P)\oplus\o(Q)$,
\item $w(P\sqcup Q) = \w(P)\oplus\w(Q)$.
\item $h(P\sqcup Q) = \max(\h(P),\h(Q))$.
\end{itemize}
\end{lemma}

\begin{lemma}[Lexicographic sum \cite{dzamonja2020}]
Let $P,Q$ be two wpos. The \emph{lexicographic sum} $P+Q$ is the disjoint union of $P$ and $Q$ ordered such that for all $p \in P, q\in Q$, $p\leq q$.
Ti is a wpo and: \begin{itemize}
\item $\o(P+ Q) = \o(P)+\o(Q)$,
\item $\w(P+ Q) = \max(\w(P),\w(Q))$,
\item $\h(P+Q)= \h(P)+\h(Q)$.
\end{itemize}
\end{lemma}

\begin{lemma}[Cartesian product \cite{dejongh77}]
Let $P,Q$ be two wpos. The \emph{cartesian product} $P\times Q$, where elements are compared component-wise,
is a wpo and: \begin{itemize}
\item$\o(P\times Q) = \o(P)\otimes\o(Q)$,
\item $\h(P\times Q) = \sup \setof{\alpha\oplus\beta+1}{\alpha<\h(P),\beta<\h(Q)}$.
\end{itemize}
\end{lemma}

\begin{lemma}[Lexicographic product \cite{dzamonja2020}]
\label{prod-lex}
Let $P,Q$ be two wpos. The \emph{lexicographic product} of $P$ along $Q$, written $P\cdot Q$, has the same support as $P\times Q$, and but is ordered differently:
$$ (p,q)\leq_{P\cdot Q} (p',q') \quad\text{ iff }\quad q<_Q q' \text{ , or } q=q' \text{ and } p \leq_P p'\;.$$
It is a wpo, and: \begin{itemize}
\item $\o(P\cdot Q)=\o(P)\cdot\o(Q)$,
\item $\w(P\cdot Q)=\w(P)\odot \w(Q)$.
\end{itemize}
\end{lemma}

See the appendix for definitions of the natural sum $\oplus$ and product $\otimes$, as well as for the lesser-known Hessenberg-based product $\odot$.

\subsection{Comparing wpos}

A widely-used and intuitive relation between wpos is the \emph{reflection} relation. A mapping between wpos
	$f \colon (A, \leq_A) \rightarrow (B, \leq_B)$
	is a \emph{reflection} if $f(x)\leq_B f(y)$ implies $x \leq_A y$, i.e. it is a morphism
	from $(A, \not\leq_A)$ to $(B, \not\leq_B)$
	When there is a reflection from $A$ to $B$, we note $A \refl B$. 
	
However, in this article, we prefer to use the stronger notions of augmentations and substructures.

\begin{definition}[Substructure, augmentation]
	A wpo $(A, \leq_A)$ is a \emph{substructure} of a wpo $(B, \leq_B)$
	whenever $A \subseteq B$ and $\leq_A$ is the restriction of $\leq_B$
	to $A$. This relation is written $A \leqstruct B$.
	Similarly $(A,\leq_A)$ is an \emph{augmentation} of $(B,\leq_B)$
	whenever $A = B$ and ${\leq_B} \subseteq {\leq_A}$.
	We write this relation $A \geqaug B$.

\end{definition}
Obviously, $A \leqstruct B$ or $A \geqaug B$ imply $A \refl B$.

We denote with $A\equiv B$ that $(A, \leq_A)$ is isomorphic to $(B, \leq_B)$.

We often abuse these notations and write $A \leqstruct B$ (resp. $B \leqaug A$) to mean that $A$ is isomorphic to a substructure (resp. an augmentation) of $B$.

\begin{example}
It is well-known that the height and the \mot of a wpo $X$ are reached by the order type of a maximal chain and a maximal linearisation of $X$, respectively. Therefore $\h(X)\leqstruct X$ and $\o(X)\geqaug X$.
\end{example}

In this article, when we consider a subset $Y$ of a wpo $X$, it is understood that $Y\leqstruct X$, i.e. $Y$ is ordered with $\leq_X$ restricted to the subset.

These notions allow us to compare the ordinal invariants of wpos.
\begin{lemma}
Let $A$ and $B$ be wpos.

If $A \leqstruct B$ then $\inv(A)\leq\inv(B)$ for $\inv\in\set{\o,\h,\w}$.

If $A \geqaug B$ then $\o(A)\leq\o(B)$ and $\w(A)\leq\w(B)$. However $\h(A)\geq\h(B)$.
\end{lemma}

The substructure and augmentation relations are monotonous through most operations on wpos.

\subsection{Orderings on the set of finite multiset}

We note $\mset{x_1,\dots, x_n}$ the finite multiset that contains the elements $x_1,\dots,x_n$. The symbols $\cup,\cap$ and $\setminus$ denote the union, the intersection and the subtraction, respectively. For any $n\in\mathbb{N}$, $m\times n$ means the union of $n$ copy of $m$.
For instance, let $m=\mset{1,1,1,2,2,3}$ and $m'=\mset{1,1,2,4}$.  Then $m\cup m'= \mset{1,1,1,1,1,2,2,2,3,4}$, and $m\cap m'= \mset{1,1,2}$. Furthermore, $m\setminus (m\cap m')=\mset{1,2,3}$, and $m'\times 2 = \mset{1,1,1,1,2,2,4,4}$

We denote with $|m|$ the number of elements of a multiset $m$.

There are two main orderings classically defined on the set of finite multisets $M(X)$ of a partial order $X$:

\begin{definition}[multiset embedding]
$\Mul(X)=(M(X),\leq_{\diamond})$ is ordered with the \emph{multiset embedding}, also known as the term ordering:

$m \leq_{\diamond} m'$  iff there exists $f:m\rightarrow m'$ injective such that for any $x\in m$,  $x\leq f(x)$.
\end{definition}
\begin{definition}[Multiset ordering]
$\Mulr(X)=(M(X),\leq_r)$ is ordered with the \emph{multiset ordering}:
$$m \leq_{r} m' \iff m=m' \text{ or } \forall x \in m\setminus(m\cap m'),\exists y \in m'\setminus(m\cap m'), x<y\;.$$
\end{definition}

The multiset ordering and the multiset embedding are both augmentations of the word embedding on $X^*$. Therefore, according to Higman's lemma \cite{higman52}, $\Mul(X)$ and $\Mulr(X)$ are wpos when $X$ is.

The main differences between these two ordering are that the multiset embedding gives a lot of importance to the size of a multiset (for any multisets $m,m'$, $m\leq_{\diamond} m'$ implies $|m| \leq |m'|$), whereas with the multiset ordering one element can dominate many. For instance, $\mset{1}\geq_r \mset{0}\cdot 100$ in $\Mulr(\mathbb{N})$.   
This is why for any wpo $X$, $\Mul(X)\leqaug \Mulr(X)$, as was observed by Aschenbrenner and Pong \cite{aschenbrenner2004}.

Note that if $X$ is a linear ordering, then $\Mulr(X)$ remains linear, while $\Mul(X)$ does not.

Nonetheless, $\leq_{\diamond}$ and $\leq_r$ behave similarly on simple data structures, as we will see with the 
following \emph{transformation lemmas}:

\begin{lemma}[Disjoint sum transformation]
\label{lem-mul-sqcup}
For any wpos $A$ and $B$,
$M^*(A\sqcup B) \equiv M^*(A)\times M^*(B)$  for $*\in\{\diamond, r\}$.
\end{lemma}

\begin{lemma}[Direct sum transformation, multiset ordering]
\label{lem-mulr-plus}
For any wpos $A$ and $B$,
$\Mulr(A+B)\equiv \Mulr(A)\cdot \Mulr(B)$
\end{lemma}

\begin{lemma}[Direct sum transformation, multiset embedding]
\label{lem-mul-plus}
For any wpos $A$ and $B$,
$\Mul(A+B)\leqaug \Mul(A)\cdot \Mul(B)$
\end{lemma}
\begin{proof}
For all $m\in \Mul(A+B)$, $m=m_A\cup m_B$ for some $m_A\in \Mul(A)$, $m_B\in \Mul(B)$. 
Let $m=m_A\cup m_B, m' = m'_A\cup m'_B$ such that $m  \leq_{\Mul(A+B)}  m'$. By definition there exists an injective function $f: m \rightarrow m'$ such that $x \leq_{A+B} f(x)$ for any $x\in m$. Then $f(m_B)\subseteq m'_B$ hence $m_B \leq_{\Mul(B)} m'_B$. If 
$m_B <_{\Mul(B)} m'_B$ then $(m_A,m_B)<_{\Mul(A)\cdot \Mul(B)}(m'_A,m'_B)$. Otherwise $m_B = m'_B$. Then $f(m_A)\subseteq m'_A$, hence $m_A \leq_{\Mul(A)} m'(A)$. Thus $(m_A,m_B)\leq_{\Mul(A)\cdot \Mul(B)}(m'_A,m'_B)$.
\end{proof}

The augmentation and substructure relations are monotone with respect to the multiset ordering and multiset embedding:
\begin{lemma}
Let $A,B$ be two wpos. Then $A\leqstruct B$ implies $\Mul(A)\leqstruct \Mul(B)$ and $\Mulr(A)\leqstruct \Mulr(B)$.
Moreover, $A \geqaug B$ implies that $\Mul(A)\geqaug \Mul(B)$ and $\Mulr(A)\geqaug \Mulr(B)$
\end{lemma}

\subsection{Ordinal invariants of the set of finite multisets}

Van der Meeren, Rathjen and Weiermann computed the \mot of $\Mul(X)$ and $\Mulr(X)$. 

\begin{theorem}[M.o.t. of the multiset embedding \cite{vandermeeren2015,weiermann2009}]
\label{thm-mul-o}
Let $X$ be a wpo. 

Then $\o(\Mul(X))=\om^{\widehat{\o(X)}}$, where $\widehat{\alpha} \eqdef \om^{\alpha'_1}+ \dots + \om^{\alpha'_n}$
for all $\alpha = \om^{\alpha_1}+ \dots + \om^{\alpha_n}$, with $\alpha_i'$ is $\alpha_i + 1$ when $\alpha_i$ is the sum of an $\epsilon$-number and a finite ordinal, otherwise $\alpha_i'=\alpha_i$.
\end{theorem}

\begin{restatable}[M.o.t. of the multiset ordering \cite{vandermeeren2015,weiermann91}]{theorem}{oMulr}
\label{thm-mulr-o} Let $X$ be a wpo.

 Then
$\o(\Mulr(X))=\om^{\o(X)}$.
\end{restatable}

Observe that $\om^{\o(X)}\leq\om^{\widehat{\o(X)}}$, as one would expect since $\Mulr(X)\geqaug \Mul(X)$. Furthermore, we should have $\w(\Mulr(X))\leq\w(\Mul(X))$, while $\h(\Mulr(X))\geq\h(\Mul(X))$. Unfortunately, only the height of $\Mul(X)$ is known for now.

\begin{theorem}[Height of the multiset embedding \cite{dzamonja2020}]
\label{thm-mul-h}
Let $X$ be a wpo. 

Then $\h(\Mul(X))=\h^*(X)$, where
 \begin{equation*}
		\h^*(A) \eqdef
		\begin{cases}
			\h(A) &\text{ if } \h(A)\text{ is infinite and indecomposable, } \\
			\h(A)\cdot\om           & \text{ otherwise.}
		\end{cases}
	\end{equation*}
\end{theorem}

\subsection{A tool to compute the width: Quasi-incomparable subsets}

Of all three ordinal invariants, the width is the less studied, since it has been introduced more recently. It is also the hardest invariant to study: unlike the \mot and the height that are reached by the order types of the maximal linearisation and the maximal chain of a wpo, respectively, the width is not embodied by a sort of "maximal antichain".

A powerful tool for analysing the width of a wpo is the notion of \emph{quasi-incomparable} subsets of a wpo, which was first used in \cite{arxiv2202.07487} for the cartesian product of several ordinals.

For any partial order $X$, we denote with $x\perp y$ when two elements $x,y\in X$ are incomparable. This notation can be extended to subsets $Y,Z$ of $X$: we denote with $Y\perp Z$ when for every $y\in Y, z\in Z$, $y\perp z$.

\begin{definition}
Let $A$ be a wpo, and $A_1,\dots,A_n$ be $n$ subsets of $A$. Then $(A_i)_{i\leq n}$ is a \emph{quasi-incomparable} family of subsets of $A$ iff for any $i<n$, for any finite $Y\subseteq A_1\cup\dots\cup A_i$, there exists $A_{i+1}'\subseteq A_{i+1}$ such that $A_{i+1}'\perp Y$ and $A_{i+1}'\equiv A_{i+1}$.
\end{definition}

This definition is slightly more restrictive than the one in \cite{arxiv2202.07487}, which only required that $\w(A_{i+1}')=\w( A_{i+1})$.

The idea behind these quasi-incomparable subsets is that sometimes one can slice a wpo $A$ into simpler subsets $A_1,\dots,A_n$ whose width is known, such that $\Inc(A_n)+\dots + \Inc(A_1)$ is embedded in $\Inc(A)$. Intuitively, it means that one can combine antichains of $A_1,\dots,A_n$ into one antichain of $A$.

This entails a practical relation between the widths of $A$ and its subsets:

\begin{lemma}
\label{lem-qi}
Let $(A_i)_{i\leq n}$ be a quasi-incomparable family of subsets of $A$. Then $\w(A)\geq \w(A_n)+\dots+\w(A_1)$.
\end{lemma}

\cref{lem-qi} will prove very useful soon, in the proofs of \cref{lem-succ,lb-mr}.

\section{Width of the multiset embedding}
In this section we compute the width of $\Mul(X)$ for any wpo $X$, which happens to be functional in the width of $X$:
\begin{restatable}[Width of the multiset embedding]{theorem}{widthMul}
 \label{thm-w-Mul}
For any wpo $X$, $\w(\Mul(X))=\om^{\widehat{\o(X)}-1}$.

(See \cref{thm-mul-o} for a definition of $\widehat{\alpha}$.)
\end{restatable}

We already know that, in some cases, the width of the multiset embedding reaches its m.o.t.

\begin{lemma}[\cite{dzamonja2020}]
\label{thm-mul-w-equal-o-indecomposable}
If $\o(X)$ is infinite indecomposable, $\w(\Mul(X))=\o(\Mul(X))$.
\end{lemma}
\begin{proof}
According to \cref{thm-mul-o}, $\o(\Mul(X))=\om^{\widehat{\o(X)}}$, which is multiplicatively indecomposable when $\o(X)$ is indecomposable. Furthermore, when $\o(X)>1$, $\o(\Mul(X))>\h^*(X)=\h(\Mul(X))$ according to \cref{thm-mul-h}. Hence $\w(\Mul(X))=\o(\Mul(X))$ according to \cref{kt}.
\end{proof}

We focus for now on the set of finite multisets on an ordinal. Let us treat first the case of successor ordinals.

\begin{lemma}
\label{lem-succ}

For any successor ordinal $\alpha = \beta + 1$, $\w(\Mul(\alpha))\geq \w(\Mul(\beta))\cdot \om$.
\end{lemma}
\begin{proof}
We write $\Mul_{>k}(X)$ for the substructure of $\Mul(X)$ where multisets have strictly more than $k$ elements.
 According to \cref{Resw}, 
\[
\w(\Mul(\alpha))=\sup \;\{\w(\Mul(\alpha)_{\perp m})+1\;|\;m\in \Mul(\alpha)\}\;.
\]
Let us compute $\w(\Mul(\alpha)_{\perp m_n})$ with $m_n=\mset{\beta} \times n$ for any $n\in\mathbb{N}$. 
Let $M_k \eqdef \setof{\mset{\beta} \times (n-k) \cup m}{ m \in \Mul_{>k}(\beta)}$ for $k\in[1,n]$ be a family of subsets of $\Mul(\alpha)$. These subsets are relevant because, for any $k\in [1,n]$, $M_k\equiv \Mul(\beta)$, and for all $m \in \M_k$, $m\perp m_n$ since $|m|>|m_n|$.
Moreover, $(M_k)_{k\in[1,n]}$ is a quasi-incomparable family of subsets of $\Mul(\alpha)_{\perp m_n}$: for any $i<n$, for any finite $Y\subset M_1 \cup \dots \cup M_i$, let $s(Y) = max \{|m|, m\in Y\}$. Observe that $M_{i+1}$ contains $M_{i+1}\cap \Mul_{>s(Y)}(\beta)$ which is incomparable to $Y$, and isomorphic to $M_{i+1}$, which is also isomorphic to $\Mul(\beta)$ .

Therefore, according to \cref{lem-qi}, $\w(\Mul(\alpha)_{\perp m_n})\geq \w(\Mul(\beta))\cdot n$. Hence \cref{Resw} entails:
\begin{align*}
w(\Mul(\alpha))&\geq \sup \{\w(\Mul(\alpha)_{\perp m_n})+1\;|\;n\in \mathbb{N}\}\\
&\geq \w(\Mul(\beta))\cdot \om\;.
\end{align*}
\end{proof}

\begin{lemma}
\label{lem-mul-alpha}
For any infinite ordinal $\alpha$, $\w(\Mul(\alpha))=\o(\Mul(\alpha))$.
\end{lemma}
\begin{proof}
We already know that $\w(\Mul(\alpha))\leq\o(\Mul(\alpha))$.
We prove the lower bound by induction on $\alpha$:

\begin{itemize}

\item If $\alpha$ is infinite indecomposable, see \cref{thm-mul-w-equal-o-indecomposable}. 

\item If $\alpha= \beta + 1$ , then according to \cref{lem-succ},
\begin{align*}
\w(\Mul(\alpha)) &\geq \w(\Mul(\beta))\cdot \om\\
 &\overset{IH}{=} \o(\Mul(\beta))\cdot\om\\
 &=\om^{\widehat{\beta}+1}=\om^{\widehat{\beta+1}}=\o(\Mul(\alpha)) \text{ according to \cref{thm-mul-o}.}
\end{align*}
\item If $\alpha = \beta + \om^{\rho}$ with $\beta,\om^{\rho}<\alpha$ and $\rho>0$, then according to the transformation lemma \ref{lem-mul-plus}, $\Mul(\alpha)\geqaug \Mul(\beta)\cdot\Mul(\om^{\rho})$. Hence according to \cref{prod-lex},
\begin{align*}
\w(\Mul(\alpha))&\geq \w(\Mul(\beta))\odot\w(\Mul(\om^{\rho}))\\
	&\overset{IH}{=} \o(\Mul(\beta))\odot \o(\Mul(\om^{\rho}))\\
	&= \om^{\widehat{\beta}}\odot \om^{\widehat{\om^{\rho}}}
	= \om^{\widehat{\alpha}}\\
	& = \o(\Mul(\alpha))\;.
\end{align*}
\end{itemize}
\end{proof}

We can now prove that \cref{thm-mul-w-equal-o-indecomposable} generalizes to any infinite $X$.
\begin{lemma}
\label{w-mul-infinite}
$\w(\Mul(X))=\o(\Mul(X))$ if $\o(X)$ is infinite.
\end{lemma}

\begin{proof}
Let $\alpha=\o(X)$. Then $X\leqaug \alpha$, hence $\Mul(X) \leqaug \Mul(\alpha)$. Thus
$$\w(\Mul(\alpha))\leq \w(\Mul(X)) \leq \o(\Mul(X))\;,$$
and $\o(\Mul(X))=\o(\Mul(\alpha))$ since $\o(\Mul(X))$ only depends on $\o(X)=\alpha$.
According to \cref{lem-mul-alpha}, $\w(\Mul(\alpha))=\o(\Mul(\alpha))$, hence $\w(\Mul(X))=\o(\Mul(X))$.
\end{proof}

We can also compute the width of $\Mul(X)$ when $X$ is a finite wpo:

For any $k<\om$, we define $\Gamma_k$ as the disjoint union of $k$ singleton wpos: $$\Gamma_k\eqdef \overbrace{1 \sqcup \dots\sqcup 1}^k\;.$$

\begin{lemma}
\label{w-mul-finite}
If $\o(X)$ is finite, then $\w(\Mul(X))=\om^{\o(X)-1}$.
\end{lemma}
\begin{proof}
Let $k=\o(X)$. Then $\Gamma_k\leqaug X \leqaug k$, hence $\w(\Mul(\Gamma_k))\geq\w(\Mul(X))\geq \w(\Mul(k))$.

Since $\Mul(\Gamma_1)=\om$, the transformation lemma \ref{lem-mul-sqcup} tells us that $\Mul(\Gamma_k)$ is isomorphic to the cartesian product $\om\times\dots\times\om$ $k$ times. In general, we do not know how to compute the width of a cartesian product, but in this special case we know that $\w(\Mul(\Gamma_k))=\om^{k-1}$ \cite{arxiv2202.07487}.

Furthermore, according to \cref{lem-succ} applied $(k-1)$ times, $\w(\Mul(k))\geq \w(\Mul(1))\cdot \om^{k-1}=\om^{k-1}$. Therefore $\w(\Mul(X))=\om^{k-1}=\om^{\o(X)-1}$.

\end{proof}

We can now prove this section's main result:
\widthMul*
\begin{proof}
If $\o(X)$ is finite, then $\widehat{\o(X)}=\o(X)$. On the other hand, if $\o(X)$ is infinite, then $\widehat{\o(X)}$ is infinite too, hence $\widehat{\o(X)}-1=\widehat{\o(X)}$. Hence $\w(\Mul(X))=\om^{\widehat{\o(X)-1}}$ follows from \cref{w-mul-finite,w-mul-infinite}.
\end{proof}

\section{Height and width of the multiset ordering}

The \mot is the only invariant we know how to compute compositionally for the multiset ordering.
\oMulr*

Since the multiset ordering of a linear ordering remains linear, we can deduce that for any ordinal $\alpha$, $\Mulr(\alpha)\equiv \om^{\alpha}$. 

For the height of $\Mulr(X)$, we have a result similar to \cref{thm-mulr-o}.

\begin{theorem}[Height of the multiset ordering]
\label{thm-h-Mulr}
Let $X$ be a wpo. 

Then $\h(\Mulr(X))=\om^{\h(X)}$.
\end{theorem}
\begin{proof}
Remember that $X\geqstruct\h(X)$, and thus $\Mulr(X)\geqstruct\Mulr(\h(X))\equiv \om^{\h(X)}$. Therefore $\h(\Mulr(X))\geq\om^{\h(X)}$.

We prove the upper bound by induction on $\h(X)$. If $\h(X)=0$ then $X=\emptyset$ and $\Mulr(X)$  only contains the empty multiset, so $\h(\Mulr(X))=1=\om^{0}$.

For any multiset $m\in\Mulr(X)$, we write $X_{<m}$ for $(\cap_{x\in m} X_{\not\geq x})\cap(\cup_{x\in m} X_{<x})$. Observe that $X_{<m}$ is a subset of $X$ such that $\h(X_m)<\h(X)$.
Now suppose $X$ is not empty and $m$ is any multiset on $X$:

\begin{align*}
\Mulr(X)_{<m}&=\bigcup_{m_1+ m_2=m, m_1\neq\emptyset} \setof{m'+m_2}{m'\in\Mulr(X_{<m_1})}\\
&\equiv \bigcup_{m_1\subseteq m, m_1\neq\emptyset} \Mulr(X_{<m_1})\;.
\end{align*}

Since $\h(X_{m_1})<\h(X)$, by induction hypothesis $\h(\Mulr(X_{<m_1}))\leq\om^{\h(X_{<m_1})}<\om^{\h(X)}$. Moreover, $\om^{\h(X)}$ is indecomposable. Hence:
$$\h(\Mulr(X)_{<m})\leq \bigoplus_{m_1+ m_2=m, m_1\neq\emptyset} \h(\Mulr(\cup_{x\in m_1} X_{<x}))<\om^{\h(X)}\;.$$

Therefore $\h(\Mulr(X))\leq\om^{\h(X)}$ according to \cref{Resh}.
\end{proof}

Unfortunately, the width of the multiset ordering is harder to compute, as $\w(\Mulr(X))$ is not functional in the ordinal invariants of $X$. The following example exhibits two wpos $X_1$ and $X_2$, with identical invariants, such that $\w(\Mulr(X_1))\neq\w(\Mulr(X_2))$.

\begin{example}
\label{ex-w-Mulr-non-functional}
Let $H=\Sigma_{n<\om}\Gamma_n$. An interesting property of $H$ is that $\w(H)=\h(H)=\o(H)=\om$.
Observe that $\w(\Mulr(H))\leq\o(\Mulr(H))=\om^{\om}$ and that $\w(\Mulr(H))\geq\w(\Mulr(\Gamma_n))=\om^{n-1}$ for all $0<n<\om$. Hence $\w(\Mulr(H))=\omom$.

Consider $X_1=H+H$ and $X_2=H+\om$, two wpos with the same ordinal invariants: $\o(X_i)=\h(X_i)=\om\cdot 2$ and $\w(X_i)=\om$ for $i\in\set{1,2}$. According to \cref{lem-mulr-plus}, $\w(\Mulr(X_1))=\w(\Mulr(H))\odot \w(\Mulr(H)) =\omom \odot \omom = \om^{\om\cdot 2}$ and $\w(\Mulr(X_1))=\w(\Mulr(H))\odot \w(\Mulr(\om)) =\omom \odot 1 = \omom$.
\end{example}

We will get around this non-functionality issue by introducing a new invariant in which the width of the multiset ordering is functional. But first we will show upper and lower bounds on width.

\begin{lemma}
\label{ub-mr} Let $X$ be a wpo. Then
$$\w(\Mulr(X))\leq \sup_{x\in X,n<\om}  \w(\Mulr(X)_{\perp\mset{x}})\otimes n + 1$$
\end{lemma}
\begin{proof}
Let us define the notation $\uperp$: For any multisets $m,m'\in\Mulr(X)$, $m\uperp m'$ iff $m
\cap m'=\emptyset$ and $m\perp m'$.

In other words, $m\uperp m'$ iff $m$ and $m'$ are disjoint and
there exists $x\in m$ such that for all $y'\in m'$, $x\not\geq y'$, and
there exists $x'\in m'$ such that for all $y\in m$, $x'\not\geq y$.
In particular $x'\not \geq x$.
Hence $m\uperp m'$ implies there exists $x\in m$ such that $\mset{x}\uperp m'$, which is equivalent to $\mset{x}\perp m'$. Now:

\begin{align*}
\Mulr(X)_{\perp m}&\geqaug \bigsqcup_{m_1+m_2= m, m_1\neq\emptyset} \setof{m'+m_2}{m'\in\Mulr(X),m'\uperp m_1}\\
&\equiv \bigsqcup_{m_1\subseteq m, m_1\neq\emptyset} \Mulr(X)_{\uperp m_1}\\
&\leqstruct\geqaug \bigsqcup_{m_1\subseteq m, m_1\neq\emptyset}\;\bigsqcup_{x\in m_1}\Mulr(X)_{\perp \mset{x}}
\end{align*}

Hence according to \cref{def-sqcup}, $$\Mulr(X)_{\perp m}\leq \bigoplus_{m_1\subseteq m, m_1\neq\emptyset}\;\bigoplus_{x\in m_1}\w(\Mulr(X)_{\perp \mset{x}})\;.$$

Let $x\in m$ such that $\w(\Mulr(X)_{\perp \mset{x}})$ is maximal.
Then $\w(\Mulr(X)_{\perp m})\leq \w(\Mulr(X)_{\perp \mset{x}})\otimes n$ for some $n<\om$.
Hence according to \cref{Resw},
$$\w(\Mulr(X))= \sup_{m\in\Mulr(X)} \w(\Mulr(X)_{\perp m})+1
\leq \sup_{x\in X,n<\om} \w(\Mulr(X)_{\perp \mset{x}})\otimes n + 1\;.$$
\end{proof}

Here is a similar lower bound:

\begin{lemma}
\label{lb-mr}
Let $X$ be a wpo. Then
$$\w(\Mulr(X))\geq \sup_{x\in X,n<\om} \w(\Mulr(X)_{\perp\mset{x}})\cdot n+1 $$
\end{lemma}
\begin{proof}
This proof follows the same structure as the proof of \cref{lem-succ}: We study the residual of $\Mulr(X)$ which contains every element incomparable to some multiset of the form $\mset{x}\times n$, and slice this residual into a family of quasi-incomparable subsets.

 According to \cref{Resw}, 
\begin{align*}
w(\Mulr(X))&=\sup_{m\in\Mulr(X)} \w(\Mulr(X)_{\perp m}) +1\\
&\geq \sup_{x\in X,n<\om}\w(\Mulr(X)_{\perp \mset{x}\times n}) + 1 \;.
\end{align*}

Let us compute a lower bound for $\w(\Mulr(X)_{\perp \mset{x}\times n})$ for any $x\in X$ and $n<\om$.
For all $k\in[1,n]$, let $M_k=\setof{\mset{x}\times (n-k) \cup m}{m\in\Mulr(X)_{\perp \mset{x}}}$. Observe that $M_k\equiv \Mulr(X)_{\perp \mset{x}}$ for any $k\in[1,n]$, and for all $m\in M_k$, $m\perp \mset{x} \times n$.
We claim that $(M_k)_{k\in[1,n]}$ is a quasi-incomparable family of subsets of $\Mulr(X)_{\perp (\mset{x}\times n)}$:
Let $i< n$ and $Y$ a finite subset of $M_1\cup\dots\cup M_{i}$. We define $m_Y$ and $M'_{i+1}$ as
\begin{align*}
m_y&\eqdef\bigcup_{j\leq i}\;\bigcup_{m\in (M_j\cap Y)} (m\setminus (\mset{x}\times (n-j)))\;,\\
M'_{i+1}&\eqdef\setof{\mset{x}\times (n-i-1)\cup m_Y \cup m}{m\in\Mulr(X)_{\perp \mset{x}}}\;.
\end{align*}
 Observe that $M'_{i+1}$ is a subset of and isomorphic to $M_{i+1}$,
and $Y\perp M'_{i+1}$.

 Therefore according to \cref{lem-qi},
 $\w(\Mulr(X)_{\perp (\mset{x}\times n)})\geq \w(\Mulr(X)_{\perp\mset{x}})\cdot n$.
\end{proof}

The bounds provided by \cref{ub-mr,lb-mr} actually match. Furthermore, they can be reformulated in such a way that the residual on $\Mulr(X)$ boils down to a residual on $X$:

\begin{theorem}
\label{thm-width-Mulr}
For any non-linear wpo $X$,
\begin{equation}
\label{eq-mulr-width}\tag{W}
\w(\Mulr(X))= \sup \setof{\w(\Mulr(X_{\not\geq x}))\cdot \om}{x\in X ,X_{\perp x}\neq \emptyset}\;.
\end{equation}

\end{theorem}
\begin{proof}
For any ordinal $\alpha$, $\sup_{n<\om} (\alpha\cdot n + 1) = \sup_{n<\om} (\alpha\otimes n + 1) = \alpha \cdot \om$.
Hence according to \cref{lb-mr,ub-mr}, $$\w(\Mulr(X))= \sup_{x\in X} (\w(\Mulr(X)_{\perp\mset{x}})\cdot \om)\;. $$
Let $x\in X$. If $X_{\perp x}=\emptyset$, then $\Mulr(X)_{\perp\mset{x}}=\emptyset$.
Otherwise let $y\in X_{\perp x}$.
Then $\setof{\mset{y}\cup m}{m\in\Mulr(X_{\not\geq x})}\leqstruct \Mulr(X)_{\perp\mset{x}}\leqstruct \Mulr(X_{\not\geq x})$.
Therefore $\w(\Mulr(X)_{\perp\mset{x}})=\w( \Mulr(X_{\not\geq x}))$ if $X_{\perp x}\neq \emptyset$, otherwise $\w(\Mulr(X)_{\perp\mset{x}})=0$.
\end{proof}

\cref{eq-mulr-width}, even though its inductive formulation makes it impractical, sparks our interest for two reasons:
First, it proves that $\w(\Mulr(X)$ is indecomposable.
Second, its structure mirrors the residual equations: it refers to a residual of $X$, and not the residual that appears in \cref{Resw} as one could expect, but the residual of \cref{Reso}. 
Furthermore, it is a supremum, indexed not on $X$ but only on elements of $X$ that are incomparable to other elements. In other words, the supremum is indexed on $X$ stripped of its "linear" elements.

These observations are what led us to the definition of fourth ordinal invariant: \emph{maximal safe order type}

\begin{definition}[Stripped subset, safe subset, maximal safe order type]
\label{def-pmot}
For any wpo $X$, $\str(X)\eqdef \setof{x\in X}{X_{\perp x}\neq\emptyset}$ is the \emph{stripped subset} of $X$.
A subset $X'$ of $X$ is \emph{safe} if there exists a maximal linearisation $\ell:X'\rightarrow \o(X')$ (i.e. a morphism) such that for all $x_1,\dots,x_n\in X'$, for all $x\in X'$ such that $\ell(x)<\ell(x_i)$ for any $i\in[1,n]$, $(X_{\not\geq x_1,\dots,x_n})_{\perp x}\neq \emptyset$. We say that $\ell$ verifies the \emph{safety condition}.

The \emph{maximal safe order type} of $X$, denoted with $\pmot(X)$, is the supremum of the m.o.t.s of safe subsets of $X$:
$$\pmot(X)\eqdef \sup\setof{\o(X')}{X'\subseteq X \text{ safe}}\;.$$
\end{definition}

Observe that $\str(\str(X))=\str(X)$, and that for any safe subset $X'$ of $X$, $X'\leqstruct \str(X)$, thus $\pmot(X)\leq\o(\str(X))$. At first glance, it looks like \cref{eq-mulr-width} could be reformulated to be functional in $\o(\str(X)$, but we will soon see that a more restricted condition like the safety condition is exactly what we need to develop this inductive expression.

\begin{lemma}
\label{pmot-reached}
For any wpo $X$, for any linearisation $\ell : \str(X) \rightarrow \o(\str(X))$, there exists a safe subset $X'$ of $\str(X)$ such that $\ell$ restricted to $X'$ verifies the safety condition, and $\o(X')\geq\delta(\o(\str(X)))$, where 
$$\delta(\alpha)\eqdef\begin{cases}
\alpha \text{ if $\alpha$ is limit,}\\
\gamma +\lfloor n/2 \rfloor \text{ if $\alpha=\gamma +n$ with $\gamma$ limit and $n<\om$.}
\end{cases}$$
\end{lemma}
\begin{proof}
Let $\downarrow \beta \eqdef \set{\gamma : \gamma<\beta}$ for any ordinal $\beta$.
We prove by induction on $\beta\leq\o(\str(X))$ that there exists $X_{\beta}\subseteq \ell^{-1}(\downarrow\beta)$ a safe subset of $\str(X)$ where $\ell$ restricted to $X_{\beta}$ verifies the safety condition, such that $\o(X_{\beta})\geq \delta(\beta)$. Furthermore, for all $\beta<\beta'\leq\o(\str(X))$, $X_{\beta}\subseteq X_{\beta'}$.

Let $X_0=\emptyset$, and $X_{\beta}=X_{\beta'}$ if $\beta=\beta'+1$ with $\beta$ limit.

If $\beta=\beta'+2$ for some $\beta'$ (it does not have to be limit), then let $x=\ell^{-1}(\beta')$ and $x'=\ell^{-1}(\beta'+1)$.
If $X_{\beta'}\cup \{x\}$ is a safe subset of $\str(X)$ then let $X_{\beta} =X_{\beta'}\cup \{x\}$.
Otherwise it means that there exists $y\in X_{\beta'}$ such that for  any $z\in \str(X)_{\perp y}$, we have $z\geq x$. 
Suppose that $X_{\beta'}\cup \{x'\}$ is not a safe subset either: it would also mean that there exists $y'\in X_{\beta'}$ such that for  any $z\in \str(X)_{\perp y}$, we have $z\geq x'$. 
We know that $x\not\geq x'$ and that $y,y'\not\geq x, x'$, and we can deduce that $y\perp x$ and $y'\perp x'$ and $x\geq y'$. We also deduce that if $x\perp x'$ then $x'\geq y$, otherwise $x\perp y'$. Thus $y\perp y'$, which contradicts the supposition.
Therefore $X_{\beta} =X_{\beta'}\cup \{x'\}$ is safe.
In both cases, $\o(X_{\beta})=\o(X_{\beta'}+1)\geq \delta(\beta')+1 \geq\delta(\beta)$ by induction hypothesis.

If $\beta$ is limit, then let $X_\beta= \cup_{\beta'<\beta} X_{\beta'}$. One can see that $X_\beta$ is safe, and that $\o(X_\beta)\geq \sup_{\beta'<\beta} \delta(\beta')=\beta$ by induction hypothesis.
\end{proof}

Therefore $\delta(\o(\str(X)))\leq \pmot(X)\leq \o(\str(X))$, hence there exists a safe subset whose \mot reaches $\pmot(X)$. We can also observe that $\pmot(X)=\o(\str(X))$ when $\o(\str(X))$ is a limit ordinal, hence this gives us a sufficient condition for when maximal safe order type reaches \mot

Here is concrete, non inductive reformulation of \cref{eq-mulr-width} which reveals how $\w(\Mulr(X))$ is functional in $\pmot(X)$.

\begin{theorem}[Width of the multiset ordering]
\label{thm-w-Mulr}
Let $X$ be a wpo. 

Then $\w(\Mulr(X))=\om^{\pmot(X)}$.
\end{theorem}
\begin{proof}

\begin{description}
\item[$(\geq)$]
Let $X'$ be a safe subset of $X$.
By induction on $\o(X')$:

When $\o(X')=0$, $\w(\Mulr(X))\geq 1$.

Let $\o(X')>0$:
According to \cref{thm-width-Mulr},
\begin{align*}
w(\Mulr(X))&= \sup \setof{\w(\Mulr(X_{\not\geq x}))\cdot \om}{x\in X ,X_{\perp x}\neq \emptyset}\\
&\geq \sup \setof{\w(\Mulr(X_{\not\geq x}))\cdot \om}{x\in X' }\;,
\end{align*}
since for any $x\in X'$, $X_{\perp x}\neq 0$.

$X'$ is safe so there exists a maximal linearisation $\ell:X'\rightarrow \o(X')$ that verifies the safety condition. Let $x=\ell(\beta)$ for some $\beta<\o(X')$. Then $\ell^{-1}(\downarrow\beta)$ is a safe subset of $X'_{\not \geq x}$; the maximal linearisation that verifies the safety condition is $l$ restricted to $\ell^{-1}(\downarrow\beta)$. Hence by induction hypothesis, $\w(\Mulr(X_{\not\geq x}))\geq \om^{\beta}$.
Therefore:
$$w(\Mulr(X))\geq \sup \setof{\om^{\beta}\cdot \om}{\beta<\o(X')}=\om^{\o(X')}\;.$$

\item[$(\leq)$] By induction on $\pmot(X)$. 

If $\pmot(X)=0$ then $X$ is linear hence $\w(\Mulr(X))=1$.

Otherwise let $\alpha=\pmot(X)$.
Assume there exists $x\in X$, with $X_{\perp x}\neq\emptyset$, such that $\pmot(X_{\not\geq x})=\alpha$. It means that there exists a safe subset $X'$ of $X_{\not\geq x}$ and a maximal linearisation $\ell:X'\rightarrow \alpha$ that verifies the safety condition. Let $\ell':X_{\not\geq x}\cup \{x\}\rightarrow \alpha+1$ such that $\ell'(x)\eqdef\alpha$ and $\ell'(y)\eqdef l(y)$ for any $y\in X'$. Then $\ell'$ is a maximal linearisation that verifies the safety condition for the subset $X'\cup \{x\}$ of $X$, which has a \mot of $\alpha+1$. Hence $\pmot(X)\geq\alpha+1$, we reached a contradiction.

For all $x\in X$ such that $X_{\perp x}\neq\emptyset$, we now know that $\pmot(X_{\not\geq x})<\alpha$. Thus $\w(\Mulr(X_{\not \geq x}))\cdot\om \leq \om^{\alpha}$ by induction hypothesis.
Hence, according to \cref{thm-width-Mulr}, $\w(\Mulr(X))\leq \om^{\pmot(X)}$.
\end{description}
\end{proof}

This opens many questions on this new invariant: we already know that it is closely bounded (\cref{pmot-reached}) and reached by the \mot of a safe subset. We also know that the maximal safe order type of a wpo reaches its \mot as soon as $\o(\str(X))=\o(X)$ is a limit ordinal. It is also quite easy to compute compositionally:

\begin{proposition}
For any non empty wpo $A,B$,
\begin{itemize}
\item $\pmot(A+ B)=\pmot(A) +\pmot(B)$,
\item $\pmot(A\sqcup B)=1 + (\o(A)-1) \oplus (\o(B)-1)$,
\item $\pmot(A\times B)\geq (\o(A)-1) \otimes \o(B)$.
\end{itemize}
\end{proposition}

\section*{Conclusion}

We have computed the width of the finite set of multisets with the multiset embedding. As for the multiset ordering, we have found a simple formula for its height, and showed its width is not functional in any of the usual ordinal invariants. Nonetheless, we introduced a new ordinal invariant, the maximal safe order type, in which the width is functional.

Investigating the maximal safe order type is a subject for further research: Are there equivalent characterisations? How does it relate to other concepts? Can it be computed compositionally ? Can we define a interesting class of wpos where maximal safe order type always coincides with \mot?

\bibliography{biblio}
\appendix
\section{On ordinal arithmetic}

We suppose well-known the notions of sum $\alpha+\beta$ and product $\alpha\cdot\beta$ on ordinals \cite{altman2017}. However, let us recall some definitions that might be less familiar to the reader.

\begin{definition}[Left subtraction]
For any ordinals $\alpha\geq\beta$, the subtraction $\alpha - \beta$ is the unique ordinal such that $\beta + (\alpha - \beta) = \alpha$.
In particular, $$\alpha - 1 \eqdef \begin{cases}
\alpha \text{ if $\alpha$ is infinite,}\\
n - 1 \text{ if } \alpha=n<\om\;.
\end{cases}$$
\end{definition}

\begin{definition}[Cantor normal form]
Any ordinal $\alpha$  can be expressed in Cantor normal form, as $\alpha = \sum_{i< n} \om^{\alpha_i}$, where $\alpha_0,\dots,\alpha_{n-1}$ are ordinals such that $\alpha_0\geq\alpha_1\geq\dots\geq \alpha_{n-1}$. This expression is unique.
\end{definition}
\begin{definition}[Natural operations]
Let $\alpha = \sum_{i< n} \om^{\alpha_i}$ and $\beta = \sum_{i< m} \om^{\beta_i}$ two ordinals in Cantor normal form. 

The \emph{natural sum}, or Hessenberg sum, $\alpha\oplus\beta$ is defined as $\gamma= \sum_{i< n+m} \om^{\gamma_i}$ with $\gamma_0\geq\dots\geq\gamma_{n+m-1}$ a reordering of $\alpha_0,\dots,\alpha_{n-1},\beta_0,\dots,\beta_{m-1}$.

The \emph{natural product} $\alpha\otimes\beta$ is defined as $\bigoplus_{i<n,j<m} \om^{\alpha_i \oplus\beta_j}$.
\end{definition}

Another product was derived from the Hessenberg sum:
\begin{definition}[Hessenberg-based product \cite{abbo99}]
\label{def-odot}
The \emph{Hessenberg-based product} $\alpha\odot\beta$ is defined inductively by the following relations:
\begin{align*}
\alpha \odot 0 &= 0\;,\\
\alpha \odot (\beta+1) &=(\alpha\odot \beta)\oplus\alpha\;,\\
\alpha\odot \beta &= \sup \setof{\alpha\odot\gamma}{\gamma<\beta}\text{ for limit $\beta$.}
\end{align*}
\end{definition}
We always have $\alpha\cdot\beta\leq\alpha\odot\beta\leq\alpha\otimes\beta$.

\begin{definition}
\label{def-epsilon-indecomposable}
An ordinal $\alpha$ is:\begin{itemize}
\item an \emph{$\epsilon$-number} when $\om^{\alpha}=\alpha$.
\item \emph{indecomposable} when for all $\beta,\delta<\alpha$, we have $\beta+\delta<\alpha$. Equivalently, indecomposable ordinals are ordinals of the form $\om^{\beta}$ where $\beta$ is any ordinal. They are sometimes called principal additive ordinals.
\item \emph{multiplicatively indecomposable} when for all $\beta,\delta<\alpha$, we have $\beta\cdot\delta<\alpha$. Equivalently, multiplicatively indecomposable ordinals are ordinals of the form $\om^{\beta}$ where $\beta$ is indecomposable. They are sometimes called principal multiplicative ordinals.
\end{itemize}
\end{definition}
\end{document}